\documentclass[journal]{IEEEtran}
\ifCLASSINFOpdf
\else
\fi
\usepackage[utf8]{inputenc}
\usepackage[english]{babel}

\usepackage[pdftex]{graphicx}
\usepackage{amsthm}

\newtheorem{thm}{Theorem}
\newtheorem{lem}[thm]{Lemma}

\theoremstyle{definition}

\usepackage{color}

\usepackage{algorithm}
\usepackage[noend]{algpseudocode}

\makeatletter
\def\BState{\State\hskip-\ALG@thistlm}
\makeatother

\algnewcommand\algorithmicswitch{\textbf{switch}}
\algnewcommand\algorithmiccase{\textbf{case}}
\algnewcommand\algorithmicassert{\texttt{assert}}
\algnewcommand\Assert[1]{\State \algorithmicassert(#1)}%
\algdef{SE}[SWITCH]{Switch}{EndSwitch}[1]{\algorithmicswitch\ #1\ \algorithmicdo}{\algorithmicend\ \algorithmicswitch}%
\algdef{SE}[CASE]{Case}{EndCase}[1]{\algorithmiccase\ #1}{\algorithmicend\ \algorithmiccase}%
\algtext*{EndSwitch}%
\algtext*{EndCase}%

 \usepackage{amssymb}
\usepackage{amsmath}

\begin{document}
%
\title{An Optimal Relay Scheme for Outage Minimization in Fog-based Internet-of-Things (IoT) Networks}
%
%

\author{Babatunji~Omoniwa,
        Riaz~Hussain,
        Muhammad Adil,
        Atif Shakeel,\\
        Ahmed Kamal Tahir,
        Qadeer Ul Hasan,
        and~Shahzad~A.~Malik
\thanks{B. Omoniwa is with the Computer Science Programme, National Mathematical Centre, 118, Abuja, Nigeria, email: tunjiomoniwa@gmail.com.}
\thanks{R. Hussain, M. Adil, A. Shakeel, A. K. Tahir, Q. U. Hasan, and S. A. Malik are with the Department of Electrical Engineering, COMSATS University, Islamabad, Pakistan, email: \{rhussain, atif\_shakeel, qadeer.hasan, smalik\}@comsats.edu.pk, adil34700@gmail.com, ahmedkamaltahir@ieee.org}
\thanks{Copyright (c) 2012 IEEE. Personal use of this material is permitted. However, permission to use this material for any other purposes must be obtained from the IEEE by sending a request to pubs-permissions@ieee.org.}}

%
%

\markboth{IEEE Internet of Things Journal,~Vol.~x, No.~x, August~2018}%
{Omoniwa \MakeLowercase{\textit{et al.}}: Optimal Relay Scheme for Fog-Based IoT Networks}
%



\maketitle

\begin{abstract}
Fog devices are beginning to play a key role in relaying data and services within the Internet-of-Things (IoT) ecosystem.
These relays may be static or mobile,
with the latter offering a new degree of freedom for performance improvement via careful relay mobility design.
Besides that, power conservation has been a prevalent issue in IoT networks with devices being power-constrained,
requiring optimal power-control mechanisms. In this paper,
we consider a multi-tier fog-based IoT architecture where a mobile/static fog node acts as an amplify and forward relay that transmits received information from a sensor node to a higher hierarchically-placed static fog device,
which offers some localized services.
The outage probability of the presented scenario was efficiently minimized by jointly optimizing the mobility pattern and the transmit power of the fog relay.
A closed-form analytical expression for the outage probability was derived.
Furthermore, due to the intractability and non-convexity of the formulated problem,
we applied an iterative algorithm based on the steepest descent method to arrive at a desirable objective.
Simulations reveal that the outage probability was improved by  62.7\% in the optimized-location fixed-power (OLFP) scheme,
79.3\% in the optimized-power fixed-location (OPFL) scheme, and 94.2\% in the optimized-location optimized-power (OLOP) scheme,
as against the fixed-location and fixed-power (FLFP) scheme (i.e., without optimization).
Lastly, we present an optimal relay selection strategy that chooses an appropriate relay node from randomly distributed relaying candidates.
\end{abstract}

\begin{IEEEkeywords}
Fog-based IoT, outage probability, steepest descent method, non-convex optimization, optimal relay selection strategy.
\end{IEEEkeywords}

%
\IEEEpeerreviewmaketitle
\section{Introduction}
\label{sec:Introduction}
\IEEEPARstart{F}{og}-based Internet-of-Things (IoT) is a distributed architecture that offers localized services, which may include storage, processing, database operation, integration, security, and management to IoT end-devices, leveraging on its proximity to the edge of the network. Fog devices are notable for supporting mobility, as well as real-time processing of data and service requests from a wide variety of IoT end-devices. For example, a mobile smartphone, smart wrist-watch, or an industrial robot may become a fog device to provide local control and application data analytics to IoT end-devices within any cyber-physical system (smart transport, smart grid, smart health, smart building, etc.)~\cite{Pan18}. Recently, Industry 4.0, also known as the Industrial Internet-of-Things (IIoT), which is the computerization of the manufacturing sector (mining, textile, food, factory, etc.)~\cite{Aazam2018}, has witnessed rapid technological advancements with fog devices (robots) playing a vital role. With the increase in machine-type communication (MTC), the magnanimity of real-time data produced by the increasingly large number of IoT sensors will require prompt response and processing at the network edge~\cite{Chiang16}-\cite{Yu17}. As such, the fog-based IoT has the potential to enhance the overall performance of the IoT ecosystem, by leveraging its ``things'' proximity and distributed architecture to offer efficient connectivity within the network~\cite{Liu2018}.

With the IoT vision of interconnecting heterogeneous devices~\cite{Lin16},
the sensor nodes which are often the primary source of sensed data within an environment,
communicate wirelessly with other IoT devices, usually fog devices,
in order to transmit data or control messages to a target destination.
A drawback to this communication setup is that these sensors are power-constrained, and are often isolated due to sparse deployment.
Moreover, these sensor nodes  may at any point in time experience degraded quality of service
(QoS) due to long distances or obstructions from the target destination (static fog node) as seen in Fig. \ref{fogarch}.
Due to the possibility of having no line-of-sight (LOS) communication link between the sensor node and the target destination,
relays can be a viable and valuable alternative in ensuring uninterrupted communication.

Over the past decades, there have been many other applications of relays in the area of wireless communications,
some in wireless sensor networks~\cite{Etezadi2012}, vehicular delay tolerant networks (VDTN)~\cite{Bouk2015},
unmanned aerial vehicles (UAV) communication networks~\cite{Geng2017},~\cite{Zhang2018}, and
relay-assisted D2D communications~\cite{Liu2017}, to mention but a few.
Most of these works have a singular objective of improving system performance by minimizing communication outages within the network.
Relay concept is quite new, but becoming more prominent in IoT networks
to assist in long-distance communication~\cite{Lv2018}, with the fog playing a major role.
However, mobile fog/edge nodes are often energy-constrained, unlike the static fog nodes that can be connected to a power source.
With this bottleneck,
it becomes necessary to implement an optimal power-control mechanism in order to minimize communication outage, as well as, utilize device power in a more efficient manner.

In reality, IoT sensors are sparsely dispersed within rooms and corridors in smart building applications,
on roads and vehicles in smart transport, on patients and health facilities in smart health-care, and embedded in meters and home appliances in smart grid applications. These IoT sensor nodes may at any point in time transmit data or service-requests to a destination node via a fog relay. The success of future IoT networks will greatly depend on the way "things" are managed at multi-tier fog layers, where localized service provision can be done. A key motivation for this work is the Industry 4.0, where fog devices (industrial robots) may be used to offer intermediary services, as well as convey important information between smart industrial devices. Also, surveillance drones used in agricultural plantations, industries, car parks and militarized zones for monitoring (intelligence gathering), can also be used to relay information within the network. This surveillance drone has the ability to actively change its location to improve communication performance. This paper aims at emphasizing the importance of relays in future IoT networks.

It is pertinent to note that transmitted packets are often lost or corrupted due to error-prone links, making relay a viable choice.
Another critical thing to consider, when deploying an IoT network, is minimizing the outage by providing redundant links via potential relay nodes.
In a network with several active relaying candidates, an optimal choice of a relay node can significantly raise the
efficiency of the network and consequently the target performance metric at the destination~\cite{Etezadi2012}.
Recently, research efforts have intensified in the development of relay selection strategies~\cite{Mohammed2013}.
Intuitively, a relay selection strategy is one that picks an optimal relay based on some defined objective such as minimizing the
outage probability.

\begin{figure}[!t]
\centering
\includegraphics[width=3in]{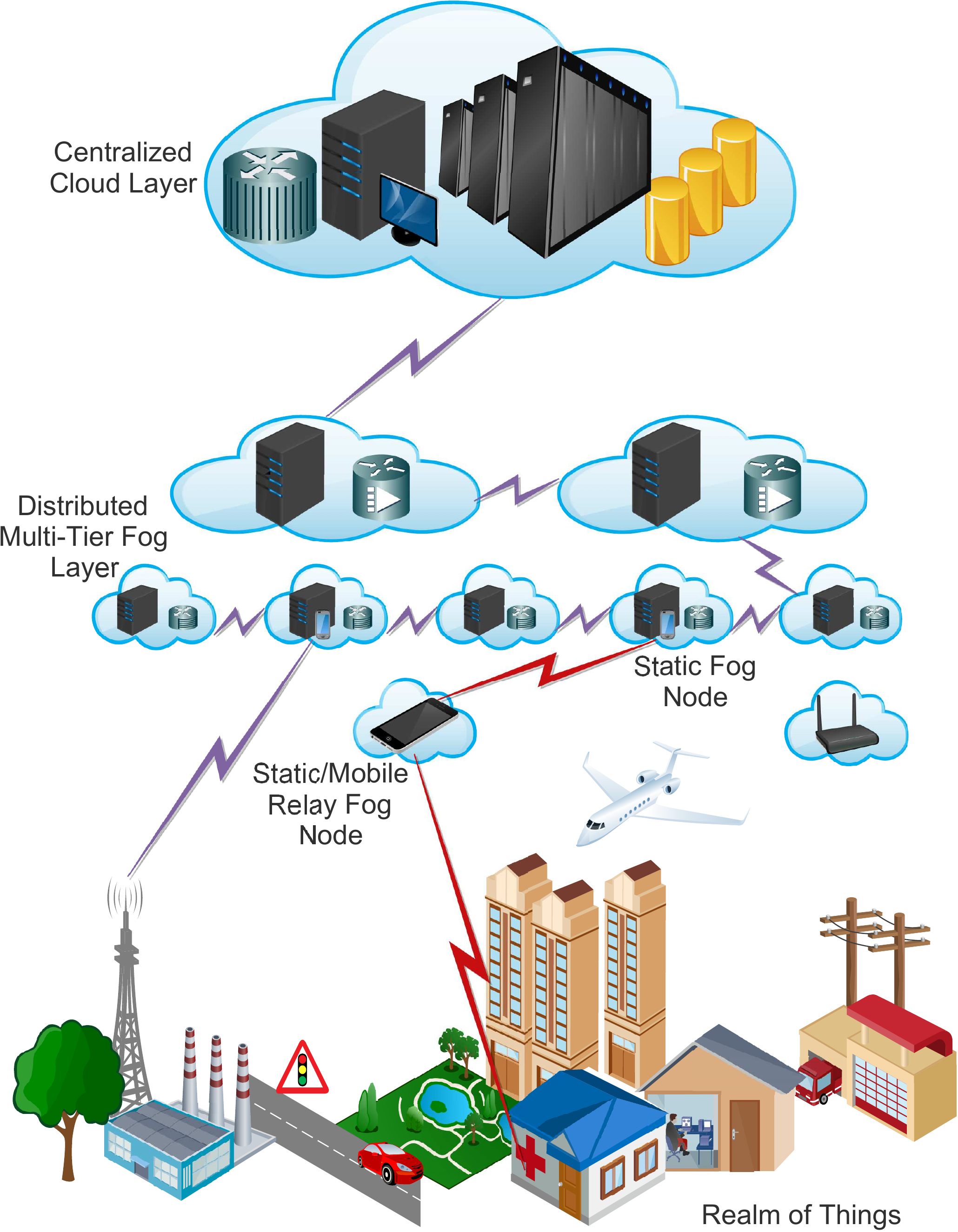}
\caption{Multi-tier fog architecture depicting a relay scenario within the IoT ecosystem.}
\label{fogarch}
\end{figure}

\subsection{Related Work}
There are several works on the use of relays to improve communication within a network. In \cite{Etezadi2012},
a decentralized WSN protocol was proposed in which ``proper'' nodes decide to act as relays without needing any central coordinator. Derivations for the average signal-to-noise ratio (SNR) at the relays and destination
were made with the assumption that the distance between the source and node and also the links between channel terminals are totally random. However, the work considers the model where each relay transmits at an equal power and has a fixed location. The work in \cite{Bouk2015} proposes a bivious relaying scheme to minimize the outage of communicating vehicles within an uncovered area between two distant neighboring Road Side Units (RSUs). Interior point algorithm was applied to the formulated problem to find the optimum speed of the target vehicle. The derived models did not consider the effect of the channel in relay selection. A repeated relay selection (RRS) technique to enhance the packet error rate (PER) performance was proposed in~\cite{Geng2017}, where the forwarded data is split into multiple segments and a relay is reselected prior to the transmission of each segment based on the updated channel state information (CSI). This strategy, however, improves PER performance at the expense of fairness within the network. In \cite{Zhang2018}, a joint power-control and trajectory optimization scheme for UAV relay networks was proposed. In this work, the outage probability was effectively minimized using the gradient descent algorithm. However, the work did not consider relay selection strategies in situations where multiple relay candidates exist.

There is quite a large body of work done in the area of IoT as a whole, however, works pertaining to relay in IoT is in the infancy stage. The work in \cite{Liu2017} proposes an overlay-based green relay assisted D2D communications with dual batteries in heterogeneous cellular networks. The proposed resource allocation approach makes available the IoT services and minimizes the overall energy consumption of the pico-relay base stations (BSs). It maximized the green energy usage by efficiently balancing the residual green energy among the pico relay BSs, thus, saving on-grid energy. However, the relay used was a static BS with the assumption that devices within the macrocell are fully covered, without consideration for outages in communication. A study on the degree of freedom (DoF) of the circular multi-relay multiple-input multiple-output interference (CMMI) channel from both the achievable DoF and converse analysis aspects, was presented in \cite{Lv2018}. A general DoF achievability problem for this model was formulated based on linear processing techniques. The CMMI network model proposed can be considered as a basic component to construct the complex IoT networks. In this work, channel parameters were considered in model design with some mathematical derivations. However, the work only focused on performance optimization of MIMO multi-relay channels (mRCs). Also, further work may be required for practical implementation of this idea.

A framework was developed in \cite{Yan2018} for wireless energy harvesting relay-assisted
underlay cognitive networks (WEH-RCRNs) with randomly distributed spatial nodes. The reuse of unused spectrum for transmitting data and harvesting ambient energy for power supplies are devised as an effective approach for large-scale IoT deployments.
A relay selection technique in WEH-RCRNs was proposed to select a suitable relay that will aid transmission. However, the prime focus was on energy harvesting, rather than on relay communication. The work in \cite{Shokrnezhad2017} focused on resource allocation in an OFDMA-based wireless IoT network, and also included channel assignment and power-control, in satisfying the SINR requirements. The resource allocation problem was formulated as a Mixed-Integer Linear Programming (MILP) problem. Furthermore, a relay-based communication model was derived and used for relay selection. The limitation of this work is that it selects the relay closest to the source as the optimal relay, which may not be valid for all cases. In \cite{Behdad2018}, a novel relay policy on RF energy harvesting was proposed for IoT Networks. The network entities considered in the work were the source, relay, and destination. The impact of energy harvesting on metrics, such as network lifetime and delay were studied to arrive at the optimum relay policy. However, the work did not consider channel variation which is an indispensable feature in IoT networks.

In this backdrop, an optimal relay scheme is proposed that minimizes communication outage, and can be readily deployed by system developers to build a robust fog-based IoT network.

\subsection{Our Work and Contributions}
Our work takes into consideration mobility and power-control constraint in minimizing the outage probability in a fog-based IoT network.
The main contributions of this paper are summarized below:
\begin{enumerate}
  \item In our fog-based IoT relaying scenario, we employ an outage minimization technique based on the steepest descent method to solve the formulated non-convex problem, thereby considerably minimizing the outage probability from source to destination.
  \item We present an iterative algorithm that efficiently minimizes the objective in three (3) proposed schemes, (a) optimized-location optimized-power (OLOP), (b) optimized-location fixed-power (OLFP), and (c) optimized-power fixed-location (OPFL) scheme, all of which yield better performance when compared with the fixed-location and fixed-power (FLFP) scheme. The choice of the scheme to be used depends on the relay design and deployment (static or mobile).
  \item We propose an optimal relay selection strategy which is based on picking an appropriate link that best minimizes the outage with the least convergent value. This proposed strategy will ensure fairness and will also help in ensuring the longevity of power-constrained fog relays.
\end{enumerate}

\subsection{Paper Organization}
The remainder of this paper is organized as follows. In
Section \ref{SystemModelProblemFormulation}, a fog-based IoT communication system model is presented.
In Section \ref{LocPowerOptimization}, we propose an iterative algorithm based on the steepest descent method,
which we apply to different relaying scenarios to minimize our objective.
In Section \ref{OptrelayselectionStrategy}, we consider a scenario where multiple relay candidates exist within the communication space,
hence, we propose an optimal relay selection strategy to select an appropriate link. Results and discussions are presented in Section \ref{Resultsec}.
Finally, Section \ref{conclusion} concludes the paper.

\section{System Model and Problem Formulation}
\label{SystemModelProblemFormulation}
In this section, we consider a scenario where an IoT sensor node (ISN) intends to transmit sensed data to a static fog node (SFN)
 without a line-of-sight link between them due to obstructions.
As such, a relay fog node needs to be deployed to aid the communication from ISN to SFN. Unlike traditional static relaying schemes,
we use a relay fog node (RFN) with some degree of freedom, that is, having the ability to change its location.
The RFN acts as an amplify and forward relay to provide ubiquitous communication between the ISN and SFN.
In this paper, we focus on the communication aided by the RFN and ignore the possibility of a direct communication between the ISN and SFN,
with a motive to emphasize the immense benefits of relays in future IoT networks.

\begin{figure}[!t]
\centering
\includegraphics[width=3in]{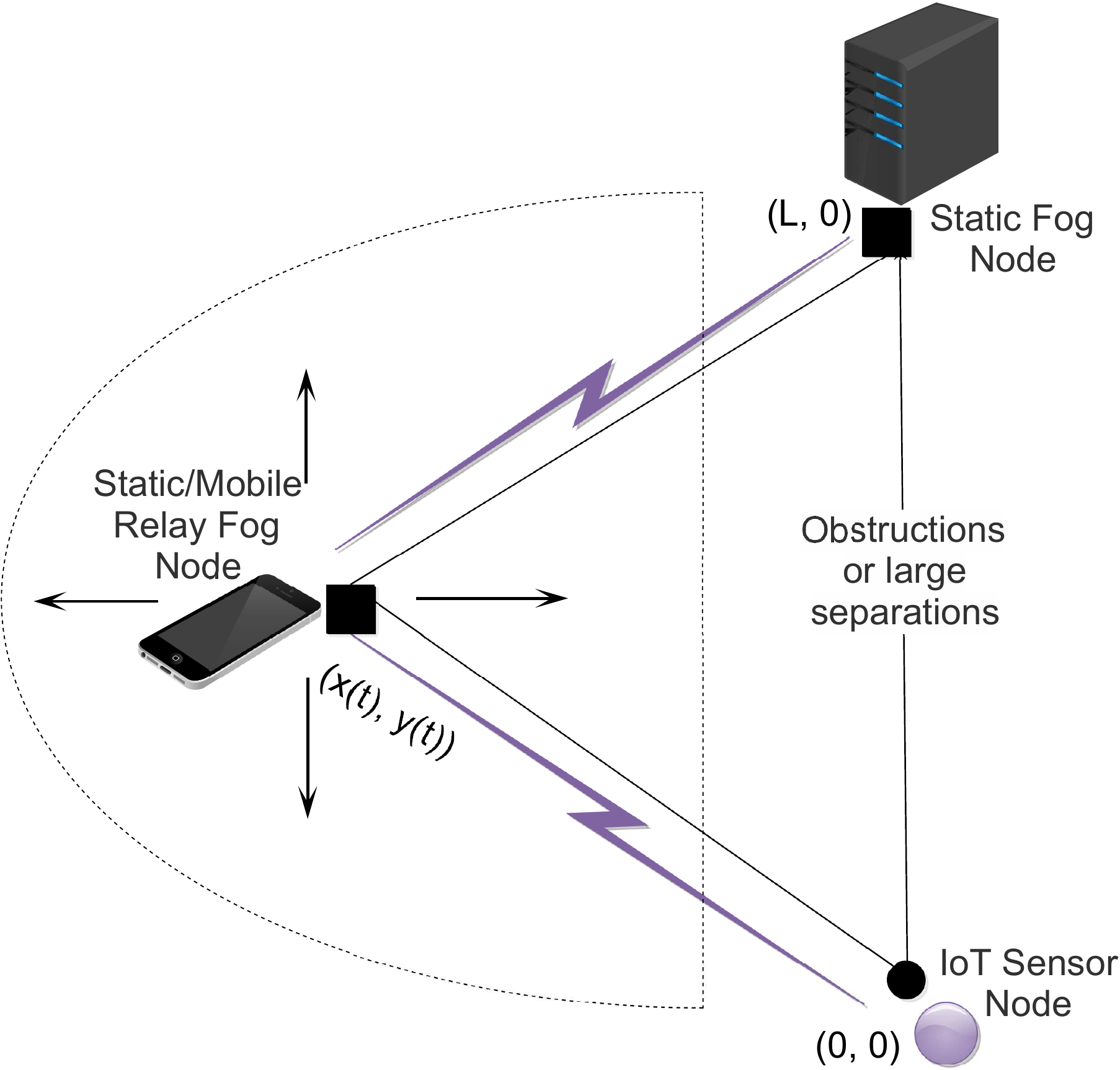}
\caption{System model for IoT communication using a static/mobile relay fog node.}
\label{systemmodel}
\end{figure}

As shown in Fig. \ref{systemmodel}, we consider a Cartesian plane with coordinates of ISN and SFN at~$(0,~0)$ and~$(L,~0)$, respectively.
We assume~$K$ time slots for the transmission phase.
The coordinates of the RFN at time slot~$t$ is given as~$(x(t),~y(t))$,~$0~\leq~t~\leq~K$.
The initial and final location of the RFN can be given as~$(x_0,~y_0)$ and~$(x_f,~y_f)$, respectively.
Using the coordinates, we can express the distance between the ISN and RFN as
\begin{equation}\label{eqn1}
  D_I = \sqrt{x(t)^2 + y(t)^2}
\end{equation}
Also, the distance between the SFN and the RFN is given as
\begin{equation}\label{eqn2}
  D_S = \sqrt{\left(x(t) - L \right)^2 + y(t)^2}
\end{equation}

During transmission, the RFN may dynamically adjust its location in order to improve quality of service (QoS).
We assume two successive time slots where the ISN transmits data to the RFN in~$t$ time slot,
then the RFN will amplify and forward the received data in~$(t + 1)$ time slot.
In both time slots, we assume that there is perfect synchronization between the communicating nodes as in~\cite{Zhang2017}. Now, the separation between the ISN and RFN  is~$D_{I}^t$ in~$t$ time slot and~$D_{I}^{t + 1}$ in~$(t + 1)$ time slot.
Also, the separation between the SFN and RFN is~$D_{S}^t$ in~$t$ time slot and~$D_{S}^{t + 1}$ in~$(t + 1)$ time slot.
We also assume a change in location of the relay in two consecutive time slots as~$\delta_{\Delta}$.
Similarly, we denote the transmit power of the ISN and the RFN to be~$P_{I}^t$ and~$P_{R}^t$, and ~$P_{I}^{t + 1}$ and~$P_{R}^{t + 1}$,
in slot time~$t$ and~$(t + 1)$, respectively. Next, to derive the outage probability,
we express the received power at the RFN from the ISN in~$t$ time slot as~\cite{Zhang2018}

\begin{equation}\label{eqn3}
  P_{R\leftarrow I}^{t} = P_{I}^{t} (D_{I}^{t})^{-\alpha} (\beta^t_I)^2
\end{equation}

We also express the received power at the SFN from the RFN in~$(t + 1)$ slot time as
\begin{equation}\label{eqn4}
  P_{S\leftarrow R}^{t + 1} = P_{R}^{t + 1} (D_{S}^{t + 1})^{-\alpha} (\beta^{t + 1}_S)^2,
\end{equation}
where~$\beta^t_I$ and~$\beta^{t + 1}_S$ are channel coefficients of ISN to RFN and RFN to SFN, respectively, modeled as~$\mathcal{N}(0,~1)$,
and~$\alpha$ is the path loss exponent.
We further express the received signal at the RFN as a function of the noise component,
which is modeled as~$\mathcal{N}(0,~N_0)$~\cite{Mohammed2013}.
\begin{equation}\label{eqn5}
  \chi_R^t = \sqrt{P_{I}^{t} (D_{I}^{t})^{-\alpha}} \beta^t_I \xi^t_I +  N_{0_R}^t,
\end{equation}
where~$\xi^t_I$ is the unit energy signal from the ISN, and the noise component that the RFN receives is given as~$N_{0_R}^t$. The RFN amplifies the received signal, hence, the resulting gain is expressed as
\begin{equation}\label{eqn6}
  G = \sqrt{\frac{P_{R}^{t + 1}}{P_{I}^{t} (D_{I}^{t})^{-\alpha} (\beta^t_I)^2  + N_0}}
\end{equation}
We can now express the received signal at the destination SFN as
\begin{equation}\label{eqn7}
\begin{split}
  \chi_S^{t + 1} & = G\sqrt{P_{I}^{t} (D_{I}^{t})^{-\alpha} (D_{S}^{t + 1})^{-\alpha}} \beta^t_I \beta^{t + 1}_S \xi^t_I \\
     & + G  N_{0_R}^t \sqrt{(D_{S}^{t + 1})^{-\alpha}} \beta^{t + 1}_S +  N_{0_S}^{t + 1},
 \end{split}
\end{equation}
where~$N_{0_S}^{t + 1}$ is the received noise at the SFN, and it is also modeled as~$\mathcal{N}(0,~N_0)$. From (\ref{eqn7}), the SNR of the network is given as
\begin{equation}\label{eqn8}
  \gamma_t = \frac{G^2 P_{I}^{t} (D_{I}^{t})^{-\alpha} (\beta^t_I)^2  (D_{S}^{t + 1})^{-\alpha} (\beta^{t + 1}_S)^2}  {G^2 N_0 (D_{S}^{t + 1})^{-\alpha} (\beta^{t + 1}_S)^2 + N_0 }
\end{equation}
In order to derive an expression for the outage probability~$\mathcal{P}_{out}$, we assume a predefined threshold~$\tilde{\gamma}$.
If the value~$\gamma_t$ given in (\ref{eqn8}) falls below the threshold~$\tilde{\gamma}$, then outage in communication occurs~\cite{Mohammed2013}.
Hence, we integrate the probability density function (PDF) of~$\gamma_t$
\begin{equation}\label{eqn9}
  \mathcal{P}_{out}^t = \mathcal{P}[\gamma_t \leq \tilde{\gamma}] = \int_0^{\tilde{\gamma}} f(\gamma_t) d\gamma_t
\end{equation}

\begin{lem}
The approximation for outage probability in successive time slots is given by
\begin{equation}\label{eqn10}
  \mathcal{P}_{out}^t = 1 - (1 + 2\Psi^2 \ln \Psi) \exp\Big( -\frac{N_0 \tilde{\gamma}}{P_{I}^{t} (D_{I}^{t})^{-\alpha}}\Big)
\end{equation}
\end{lem}
\begin{proof}
The proof is given in Appendix A.
\end{proof}
where~$\Psi = \sqrt{(N_0 \tilde{\gamma})/(P_{R}^{t + 1} (D_{S}^{t + 1})^{-\alpha})}$.
We can now formulate the optimization problem to satisfy certain constraints.
The outage probability~$\mathcal{P}_{out}^t$ given in (\ref{eqn10}) depends on the location and transmit power in the consecutive time slots. As such, our objective will be to minimize~$\mathcal{P}_{out}^t$ by optimizing the power and location of the RFN.

\begin{equation}\label{eqn10a}
\begin{split}
  ~&\min_{P_{I}^{t},~P_{R}^{t + 1},~D_{I}^{t},~D_{S}^{t + 1}} \mathcal{P}_{out}^t\\
s.t.~& P_{I}^{t} + P_{R}^{t + 1} \leq P_{max},\\
      & P_{I}^{t},~P_{R}^{t + 1} \geq 0,\\
       & \delta_{\Delta} \leq \iota,
\end{split}
\end{equation}
where~$\iota$ is the mobility constraint on the RFN.
\section{Location and Power Control Optimization}
\label{LocPowerOptimization}
Optimization techniques such as the Newton and quasi-Newton methods for solving nonlinear
problems are reputable for their speed of convergence if and only if a sufficiently accurate initial approximation is known, otherwise, these methods are limited if the initial approximation is not precise. However, the Steepest Descent Method (SDM) can be used to solve non-convex optimization problems and assures convergence to a local optimum even for poor initial approximations~\cite{Burden2011}.

The SDM is hugely important from a theoretical viewpoint, since it is one of the non-complex methods for which a satisfactory analysis exists. More advanced algorithms, like the ones proposed in this paper, are often motivated by an attempt to change the fundamental steepest descent technique in such a way that the new algorithm will have better convergence properties~\cite{Luenberger2008}.
In our work, we consider the SDM shown in Algorithm \ref{sdm}. In the SDM, the trajectory to the solution follows a zigzag pattern~\cite{Antoniou2007}.

\begin{algorithm}
\caption{Steepest Descent Method}\label{sdm}
\begin{algorithmic}[1]
\State \textbf{data} $x_0 \in R^n$
\State \textbf{initialize} $i = 0$
\BState \emph{top}:
\If {$\nabla \mathcal{P}_{out}(x_i) \leq \epsilon $} \Return stop,
\State compute search direction $h_i \gets -\nabla \mathcal{P}_{out}(x_i)$
\EndIf
\State \textbf{endif}
\BState \emph{compute the step-size}:\\
 $\lambda_i \in \arg{\min_{\lambda \geq 0}}~\mathcal{P}_{out}(x_i + \lambda h_i)$
\BState \emph{set}:\\
$x_{i + 1} = x_i + \lambda_i h_i$
\State \textbf{goto} \emph{top}.
\end{algorithmic}
\end{algorithm}

\begin{lem}
If~$\lambda$ is chosen such that~$\mathcal{P}_{out}(x_i + \lambda h_i)$ is minimized in each iteration,
then successive directions are orthogonal.
\end{lem}
\begin{proof}
The proof is given in Appendix B.
\end{proof}

In practice, several non-convex approaches, such as gradient
descent and alternating minimization, are used to arrive at a local optimum and often outperform relaxation-based approaches in terms of speed and scalability~\cite{Antoniou2007},~\cite{Boyd2013}. The relaxation-based approaches involve relaxing non-convex problems to convex ones and applying traditional methods to solve the relaxed optimization problem. To better capture the problem structure and eliminate approximation errors that may result from the relaxation-based approach, we employ a non-convex approach in this paper.
Considering the non-convex optimization problem in (\ref{eqn10a}),
we apply an iterative algorithm that is based on SDM to solve the problem in-line with the presented constraints.

Algorithm \ref{proposedalgorithm} show that our proposed solution can be applied in various schemes depending on the type of deployment the designer may desire.
Our algorithm converges to an optimal solution whenever the fractional increase in the outage probability becomes less than a predefined tolerance~$\epsilon$.
To solve for the optimal location of the RFN, we assume the power-control variable in (\ref{eqn10a}) to be given.
Thus, we now have a new expression which is also non-convex, given by
\begin{equation}\label{eqn10a1}
\begin{split}
  ~&\min_{D_{I}^{t},~D_{S}^{t + 1}} \mathcal{P}_{out}^t\\
s.t.~&\delta_{\Delta} \leq \iota
\end{split}
\end{equation}
The expressions from (\ref{eqn1}) and (\ref{eqn2}) are inserted into (\ref{eqn10a1}).
Starting with a two-dimensional initial point~$x_0$ corresponding to~$D_{I}^{t}$ and~$D_{S}^{t + 1}$,
we apply an iterative procedure to achieve local minimum outage. The step-size is a nonnegative scalar minimizing~$\mathcal{P}_{out}(x_i + \lambda h_i)$. In our work, we consider a sufficiently small step-size, such that~$0 < \lambda < \mathbb{L}$, where~$\mathbb{L}$ is the Lipschitz constant,~$(\lambda = 0.001)$ in Algorithm (\ref{sdm}), which ensures convergence to a local minima, with the gradient outage probability~$-\nabla \mathcal{P}_{out}$ in a descending direction.

\begin{lem}
The formulated problem in (\ref{eqn10a}) and (\ref{eqn10a1}) contradicts the condition for convexity, as such, it is known to be NP-hard.
\end{lem}
\begin{proof}
The function~$\mathcal{P}_{out}^t (x)$ is convex if
\begin{equation}\label{eqn10proof1}
  \mathcal{P}_{out}^t (\lambda x + (1 + \lambda)y) \leq \lambda \mathcal{P}_{out}^t (x) + (1 + \lambda) \mathcal{P}_{out}^t (y),
\end{equation}
$\forall~x, y$~and~$\forall~\lambda \in [0,1]$~\cite{Boyd2013}. Furthermore, we apply a ``D-test'' for optimizing multivariate functions, which involves solving the determinant of the Hessian
matrix of~$\mathcal{P}_{out}^t$,~$H_{f(x, y)} = \begin{pmatrix}f_{xx} & f_{xy}\\ f_{xy} & f_{yy}\end{pmatrix}$, where~$f = \mathcal{P}_{out}^t$, and~$(x_0, y_0)$ is a critical point of~$f$. Intuitively, we see that~$f$ has a continuous second partial derivative,~such that~$D(x,y) = f_{xx}f_{yy} - f_{xy}^{2}$.
\begin{itemize}
  \item If~$D(x_0, y_0) > 0$ and
  \begin{itemize}
    \item if~$f_{xx}(x_0, y_0) > 0$, then $f(x_0, y_0)$ is a local minimum of~$f$;
    \item if~$f_{xx}(x_0, y_0) < 0$, then $f(x_0, y_0)$ is a local maximum of~$f$.
  \end{itemize}
  \item If~$D(x_0, y_0) < 0$ then ~$(x_0, y_0)$ is a saddle point of~$f$.
\end{itemize}
It is important to note that there exist optimization problems that converge to saddle points in the worst case initialization of gradient descent-based methods, rather than a local minimum or local maximum. However, such worst-case analysis do not pose a serious concern to practitioners, as saddle points can easily be escaped under mild regularity conditions~\cite{Lee2016}. We apply the D-test,  to the non-convex problem in (\ref{eqn10a}), which is sufficient to show that the critical point is a local minimum.
\end{proof}

Given the location of the relay at slot~$t$, (\ref{eqn10a}) can expressed as
\begin{equation}\label{eqn10a2}
\begin{split}
  ~&\min_{P_{I}^{t},~P_{R}^{t + 1}} \mathcal{P}_{out}^t\\
s.t.~& P_{I}^{t} + P_{R}^{t + 1} \leq P_{max},\\
      & P_{I}^{t},~P_{R}^{t + 1} \geq 0
\end{split}
\end{equation}
Power-control~\cite{Zhou2015} can be efficiently achieved when~$P_{I}^{t} + P_{R}^{t + 1} = P_{max}$. When (\ref{eqn10a2}) is applied in Algorithm \ref{proposedalgorithm}, it significantly minimizes the outage probability.

\begin{algorithm}
\caption{Proposed Iterative Algorithm}\label{proposedalgorithm}
\begin{algorithmic}[1]
\State \textbf{data} $L = 50$ m,~$P_{max} = 26$ dBm
\State \textbf{initialize} $i = 0,~D_{I}^{t} = D_{S}^{t + 1} = L/2,$~$P_{I}^{t} = P_{R}^{t + 1} = P_{max}/2,~\forall~t \in \{1,~3,~5,...,~K\}$
\While{$(\mathcal{P}_{out}^t)^{ith} - (\mathcal{P}_{out}^t)^{(i-1)th} > \epsilon$}
\State $i = i + 1$

\BState \emph{optimal location}:
\For{$t = 1 : K$}\\
~~~~~apply Algorithm (\ref{sdm}) to solve (\ref{eqn10a1});
\EndFor
\State \textbf{endfor}

\BState \emph{optimal power}:
\For{$t = 1 : K$}\\
~~~~~apply Algorithm (\ref{sdm}) to solve (\ref{eqn10a2});
\EndFor
\State \textbf{endfor}

\State $s$~$\leftarrow$ DesignScheme

\Switch{$s$}
    \Case{OLFP:}\\
       ~~~~~~~~~~~~\textbf{solve} \emph{optimal location};
    \EndCase

    \Case{OPFL:}\\
       ~~~~~~~~~~~~\textbf{solve} \emph{optimal power};
    \EndCase

    \Case{OLOP:}\\
       ~~~~~~~~~~~~\textbf{jointly solve} \emph{optimal location~and optimal power};
    \EndCase

 \EndSwitch
\State \textbf{endswitch}

\EndWhile
\State \textbf{endwhile}

\end{algorithmic}
\end{algorithm}


\section{Optimal Relay Selection Strategy}
\label{OptrelayselectionStrategy}
In a network with several active relaying candidates, an optimal choice of a relaying node can substantially improve the efficiency of
the network and consequently the desired performance metric at the destination~\cite{Etezadi2012}.
We present a scenario shown in Fig. \ref{systemmodel2}, where an IoT sensor has data to send to a destination fog device with no LOS between them.
Four potential relay candidates can also be seen in the presented scenario, however,
the mechanism required to select an optimal relay will depend on metric chosen by the system designer. The RFN attempts to adjust its location, or transmit power or even both in \emph{K} time slots in order to minimize communication outage. The convergence value is attained when the outage is significantly minimized based on some pre-defined tolerance as seen in Lemma 4. In this section, we have used the convergence value as a metric to select an optimal relay (that converges fastest) from a set of potential relay candidates. We assume the convergence value to be synonymous with the counter value. Aiming to minimize the outage probability, it is intuitive to select a relay node~$R$ from a set of active fog relays~$R_{act} \in R_i$ using an optimal strategy, where~$R_i$ is the set of all relays in the fog-based IoT network.
Here, we propose a relay selection strategy that chooses the appropriate link via an optimal relay.
The selection is based on the speed of potential~$R \in R_{act}$ to converge to an optimal solution, taking into consideration mobility and power-control constraints of~$R \in R_{act}$.
Algorithm \ref{selectionalgorithm} describes the proposed relay selection strategy. The flow chart can be seen in Fig. \ref{flowchart}.

\begin{figure}[!t]
\centering
\includegraphics[width=3in]{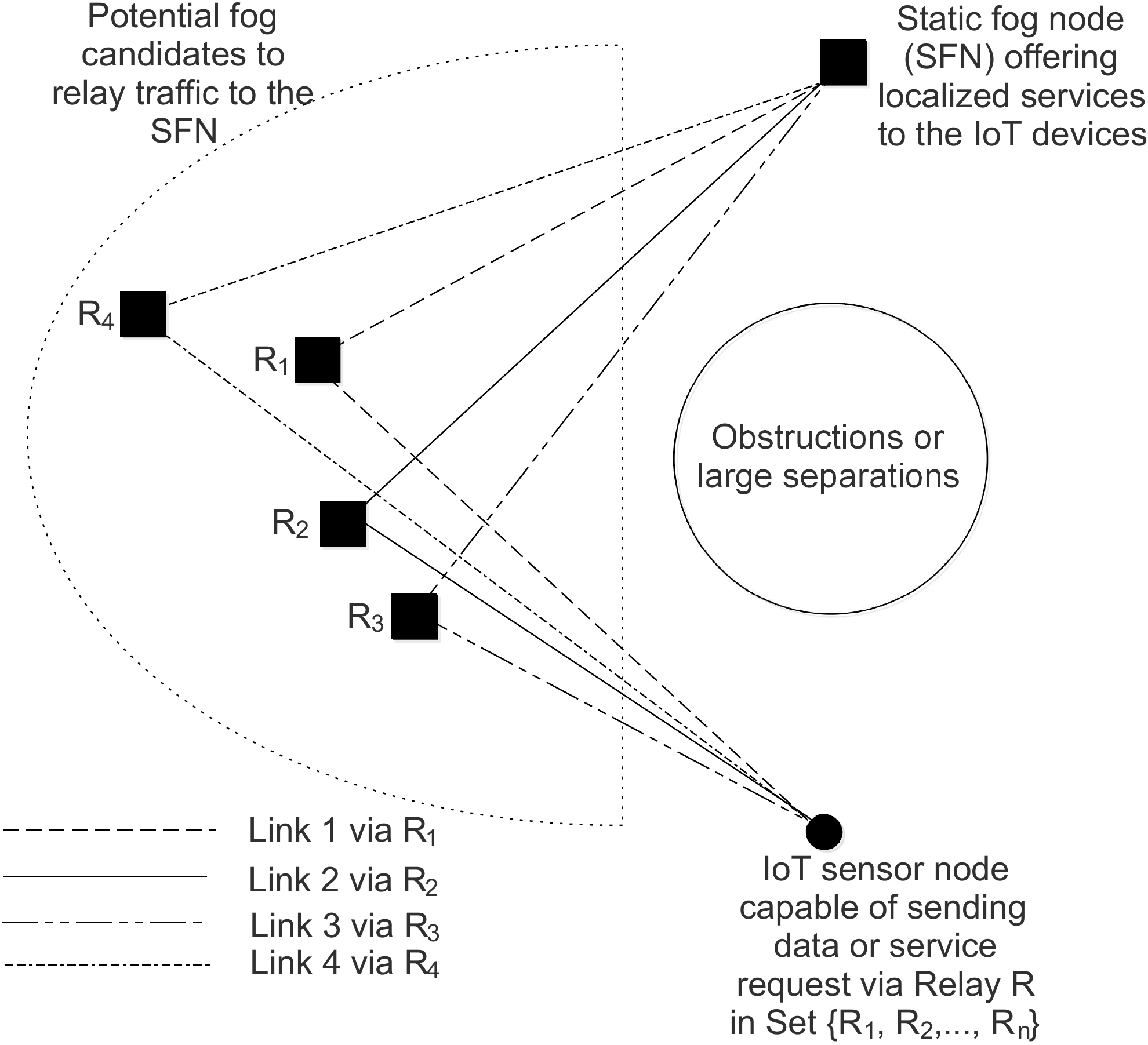}
\caption{System model for IoT communication with multiple fog relays.}
\label{systemmodel2}
\end{figure}

\begin{algorithm}
\caption{Proposed Relay Selection Strategy}\label{selectionalgorithm}
\begin{algorithmic}[1]
\State \textbf{data} $R~\in R_i,~\Theta \in \Theta_i~\forall i =1,~3,~5,...,~n$,~$R_{act}$ denotes active relays,~$\hat{R}$ denotes optimal relay
\State \textbf{initialize}~$n = 0$,~$\Theta_i = 0$;

\While{ISN has data to transmit}
\State ISN sends broadcast to the~$R_{act}$ group

\For{$R_{act} \in R_i$}
\State compute~$\Theta~\leftarrow$ Convergence value
\State Compare~$\Theta~\forall~R_{act}$ and find min~$\Theta_i$
\If {min~$\Theta_i$ is~\textit{found}~\&\&~CTS flag is received}  
\State transmit via~\textit{$\hat{R}$}
\State decrease~$\Theta$~$\forall$~$R~\neq$ \textit{$\hat{R}$}
\State recompute~$\Theta~\forall~R~\leftarrow$ \textit{$\hat{R}$}
\State \textbf{elseif} {multiple~min~$\Theta_i$ is~\textit{found}}~\&\&~CTS flag is received
\State schedule all~\textit{$\hat{R}$}~to tx in round robin manner.
\State decrease~$\Theta$~$\forall$~$R~\neq$ \textit{$\hat{R}$}
\State recompute~$\Theta~\forall~R~\leftarrow$ \textit{$\hat{R}$}
\State \textbf{else}
\State freeze~$\Theta~\forall~R~\neq$ \textit{$\hat{R}$} and wait for CTS flag
\EndIf
\State \textbf{endif}
\EndFor
\State \textbf{endfor}
\EndWhile
\State \textbf{endwhile}
\end{algorithmic}
\end{algorithm}

\begin{figure}[!t]
\centering
\includegraphics[width=3in]{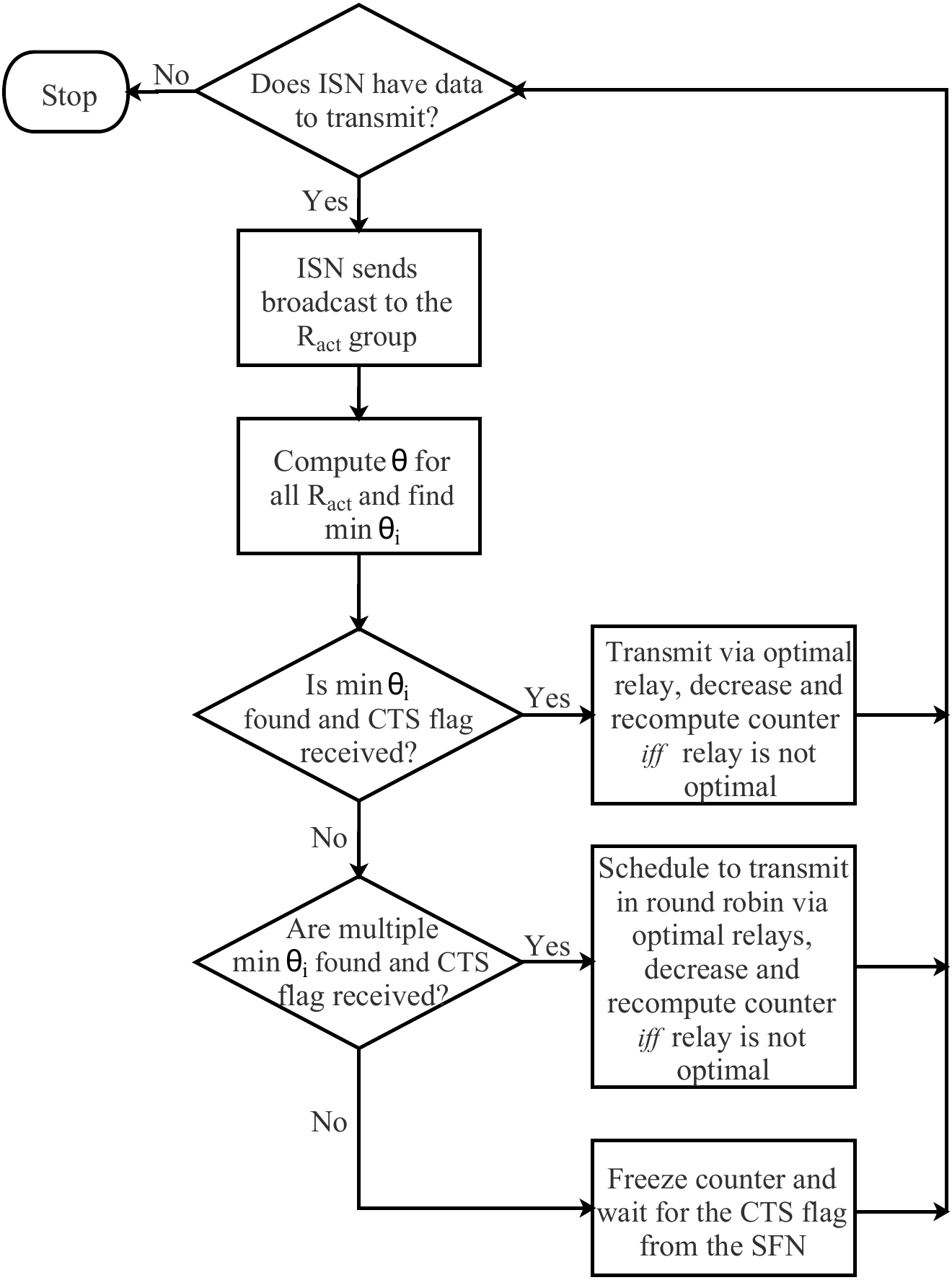}
\caption{Flow chart of the proposed relay selection strategy.}
\label{flowchart}
\end{figure}

\begin{lem}
In a fog-based IoT network with potential active relays,~$R_{act}$,
the relay~$R \in R_{act}$ that minimizes the outage probability with fastest convergence is the optimal relay~$\hat{R}$.

\end{lem}
\begin{proof}
From (\ref{eqn10}), it can be observed that the expression for outage probability improves when we apply Algorithm \ref{proposedalgorithm}. Hence, we can easily get the minimal iteration required~$\forall R \in R_{act}$ to satisfy the convergence criteria,\\~$$(\mathcal{P}_{out}^t)^{ith} - (\mathcal{P}_{out}^t)^{(i-1)th} < \epsilon$$
\end{proof}
As soon as an optimal relay~$\hat{R} \in R_{act}$ is selected, all other~$R \in R_{act}$ will deactivate their counter and wait for the next transmission phase.
Furthermore, since contention occurs only at the receiver, a clear-to-send (CTS) flag from the destination also triggers
each~$R \in R_{act}$ to begin its down-counter from the initial convergence value.
This proposed strategy will ensure fairness and will also help in ensuring the longevity of the power-constrained fog relay.
This strategy can be better pictured with an example of four active randomly distributed relaying candidates,
an IoT sensor node, and a destination fog node, as can be seen in Fig. \ref{systemmodel2}.
For example, an IoT sensor sends an initial message to probe potential relay candidates. We also assume that the ISN intermittently sends probe packets to potential relays. This action gives the ISN the knowledge of the optimal relay (with least convergence value) to be used in that transmission phase. Applying Algorithm \ref{selectionalgorithm}, given four~$R_{act}$ with convergence value~$\Theta~=~\{36,~11,~7,~63\}$~corresponding to~$R_{act}~=~\{R_1,~R_2,~R_3,~R_4\}$. Intuitively,~$R_3$ will be selected and all other~$R \in R_{act}$ that were not selected will freeze their counter and wait for a CTS flag from the destination (SFN).
This CTS flag sent by the SFN indicates that the previous transmission is complete.
Prior to receiving the CTS flag, the~$R \in R_{act}$ that were not selected can go to an idle state to minimize energy consumption.
As soon as the CTS flag is received by all~$R \in R_{act}$, they all begin to decrease their counters and wait for any data to be relayed from the IoT sensor.
However, the~$\hat{R}$ selected in the previous transmission phase will recompute a new value based on the convergence criteria.
For instance, prior to the second transmission phase when all~$R \in R_{act}$ receive a CTS flag,
we assume a case where all relay candidates decrease their counter by 5 before receiving a CTS flag, and $R_3$ which was previously selected as the optimal relay~$\hat{R}$ now recomputes its convergence time value to be 7.
The new set of convergence values will be updated as~$\Theta~=~\{31,~6,~7,~58\}$~corresponding to~$R_{act}~=~\{R_1,~R_2,~R_3,~R_4\}$.
Algorithm \ref{selectionalgorithm} stipulates that~$R_3$ which was selected as the optimal relay in the previous transmission will not decrease its counter value. The motivation for this mechanism is to allow for fairness and energy balance among potential fog relays.
In this case,~$R_2$ will be selected as~$\hat{R}$ in the second transmission phase.
In a situation where two or more~$R \in R_{act}$ have same convergence value, each will relay in consecutive transmission phase in a round robin manner. This procedure efficiently selects the optimal link with an attempt to also achieve fairness and energy balance within the IoT environment.
We aim at enhancing this strategy to become a full-fledged MAC layer protocol for IoT communication in our future work.

\section{Results}
\label{Resultsec}
In this section, the outage minimization is evaluated through simulations, which were carried out using MATLAB.
To demonstrate the efficiency of the proposed iterative algorithm using the steepest descent method,
the optimized-location optimized-power (OLOP) scheme was compared with three other schemes, the optimized-location fixed-power (OLFP),
the optimized-power fixed-location (OPFL), and the fixed-location and fixed-power (FLFP).

\begin{table}
\small
\centering
\caption{Simulation Parameters}
\label{table:parameters}
\begin{tabular}{ll}
  \hline
 \textit{Parameters} & \textit{Value} \\
  \hline \hline
   Time slots~$K$ & 1500\\
   Simulation space & 50~$\times$ 35 \emph{$m^2$}\\
   Distance~$L$ & 50 \emph{m}\\
   Maximum transmit power~$P_{max}$  & 26 \emph{dBm}\\
   SNR threshold~$\tilde{\gamma}$ & 0 \emph{dB}\\
   Path-loss exponent~$\alpha$ & 4\\
   Noise power~$N_0$ & $-96$ \emph{dBm}\\
   Tolerance~$\epsilon$ & $-10^{-2}$\\
   Mobility constraint~$\iota$  & 0.01 \emph{m}\\
      \hline \hline
 \end{tabular}
 \end{table}

\begin{figure}[!t]
\centering
\includegraphics[width=3.2in]{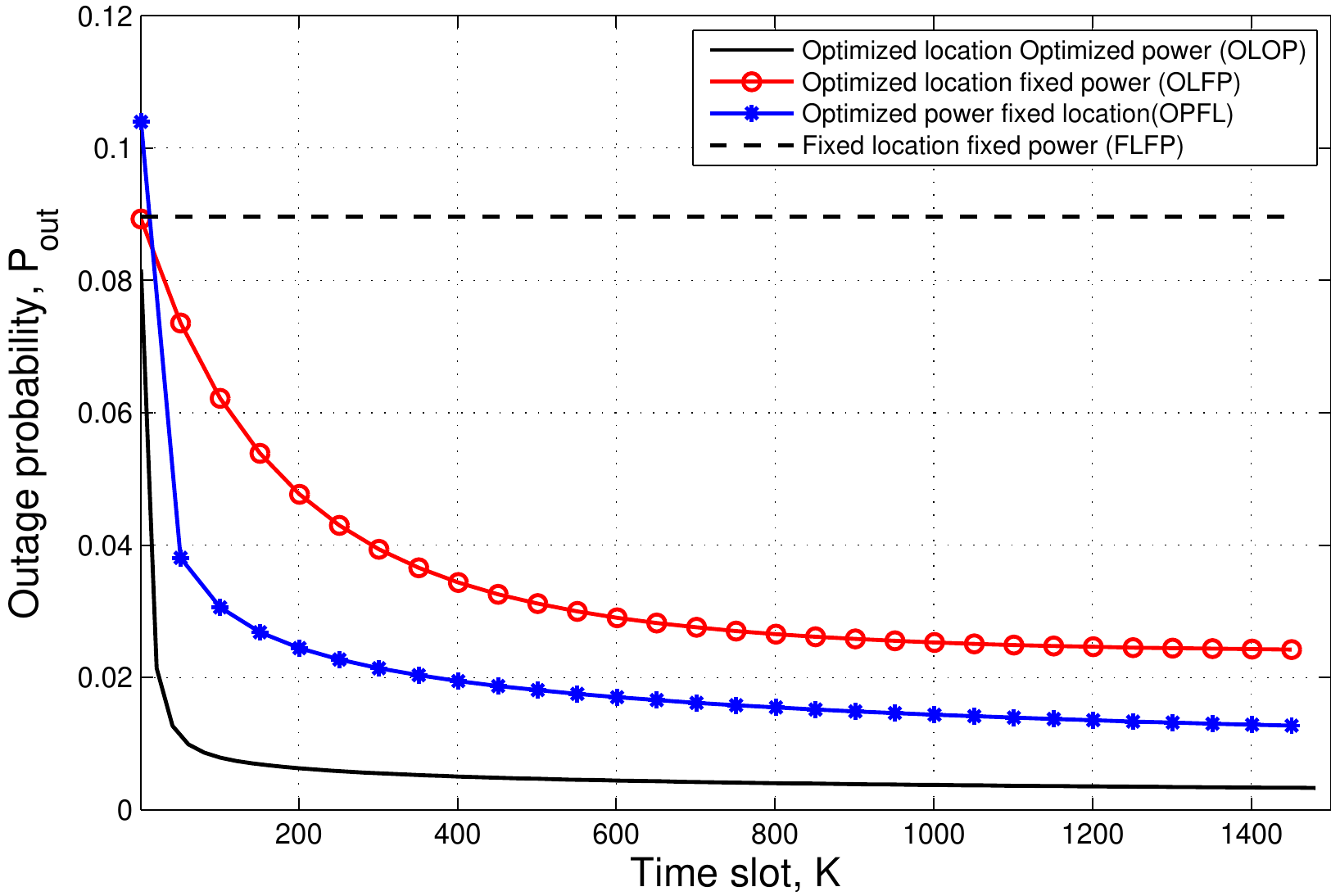}
\caption{Time slot,~$K$ vs. Outage probability, $\mathcal{P}_{out}$.}
\label{optimization}
\end{figure}

The simulation parameters are summarized in Table~\ref{table:parameters}. The transmit power~$P$, noise power~$N_0$, and other parameters used in the simulation is motivated by existing works~\cite{Joshi2017}, \cite{Sodhro2018}. In the FLFP scheme, the transmit power of the ISN and RFN is kept constant and there is also no attempt by the RFN to adjust its location for improved QoS, we observe from Fig. \ref{optimization} that the outage probability was not minimized during the entire time slot~$K$. This approach has been used in modeling several relay networks~\cite{Liu2017},~\cite{Mohammed2013}. In the OLFP scheme, the transmit powers of the ISN and MFRN are fixed at certain values, this implies that the RFN is assumed to be mobile and searches for the optimal location that will minimize the communication outage, we can observe a better improvement in the outage probability of about 62.7\% as compared to FLFP.
In the OPFL scheme, the location of the RFN is fixed with the assumption of a static relay, while the devices regularly adjust their power levels to minimize the communication outage. As seen in Fig. \ref{optimization}, the OPFL scheme gives better performance as compared with the FLFP and the OLFP, however, more power may be required to deliver a certain level of QoS. In OPLF, we observed an improvement of about 79.3\% as compared to FLFP. In the OLOP scheme, both the location and transmit power of the relay can be adjusted. From Algorithm \ref{proposedalgorithm}, the OLOP scheme jointly optimizes location and power in order to minimize the objective function. Fig. \ref{optimization} show the efficiency of OLOP as compared to all other schemes. The outage probability was minimized by 94.2\% as compared to the FLFP scheme.

\begin{figure}[!t]
\centering
\includegraphics[width=3.2in]{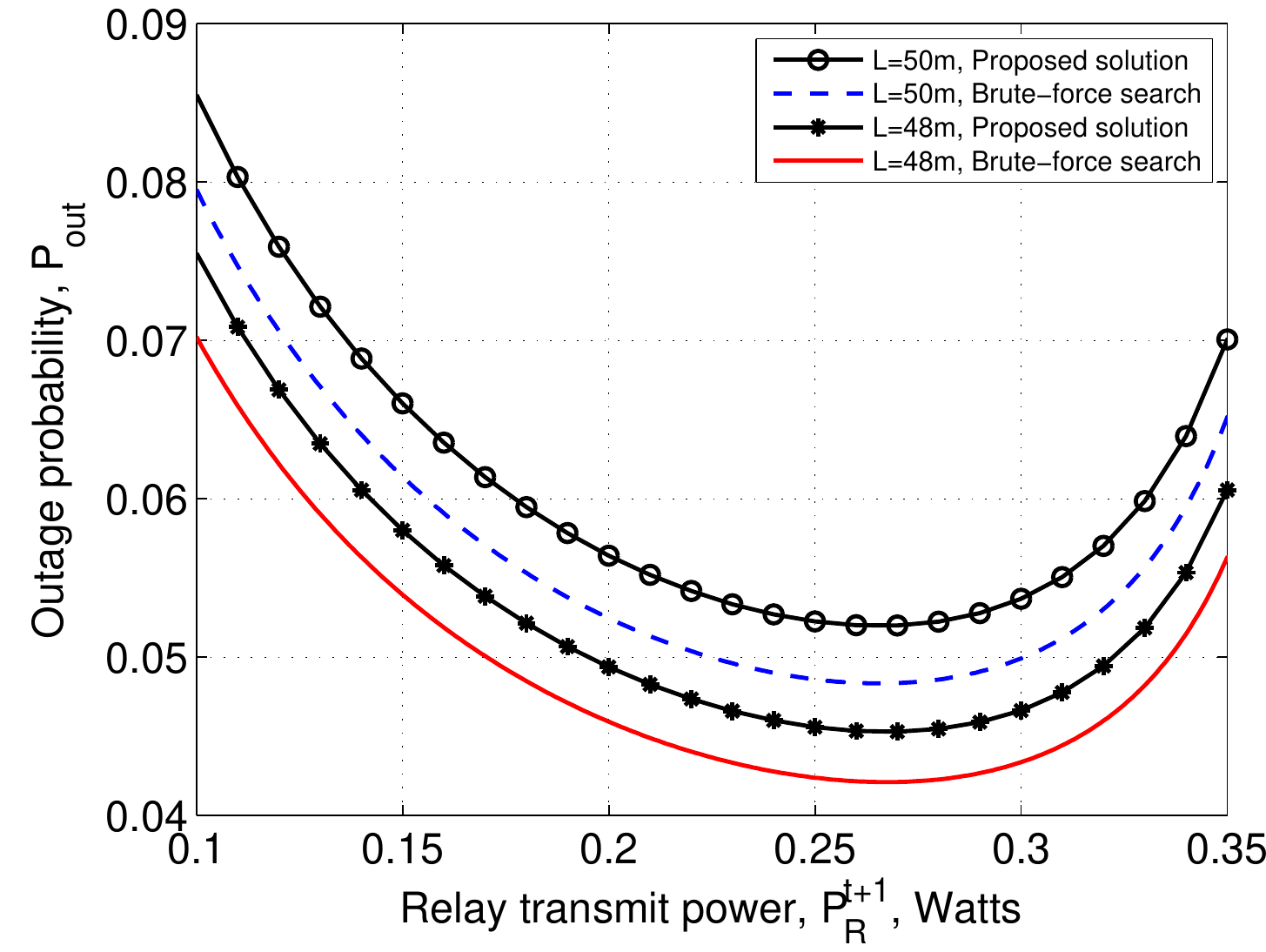}
\caption{RFN transmit power~$P_{R}^{t+1}$ vs. Outage probability, $\mathcal{P}_{out}$.}
\label{relaypower}
\end{figure}

\begin{figure}[!t]
\centering
\includegraphics[width=3.2in]{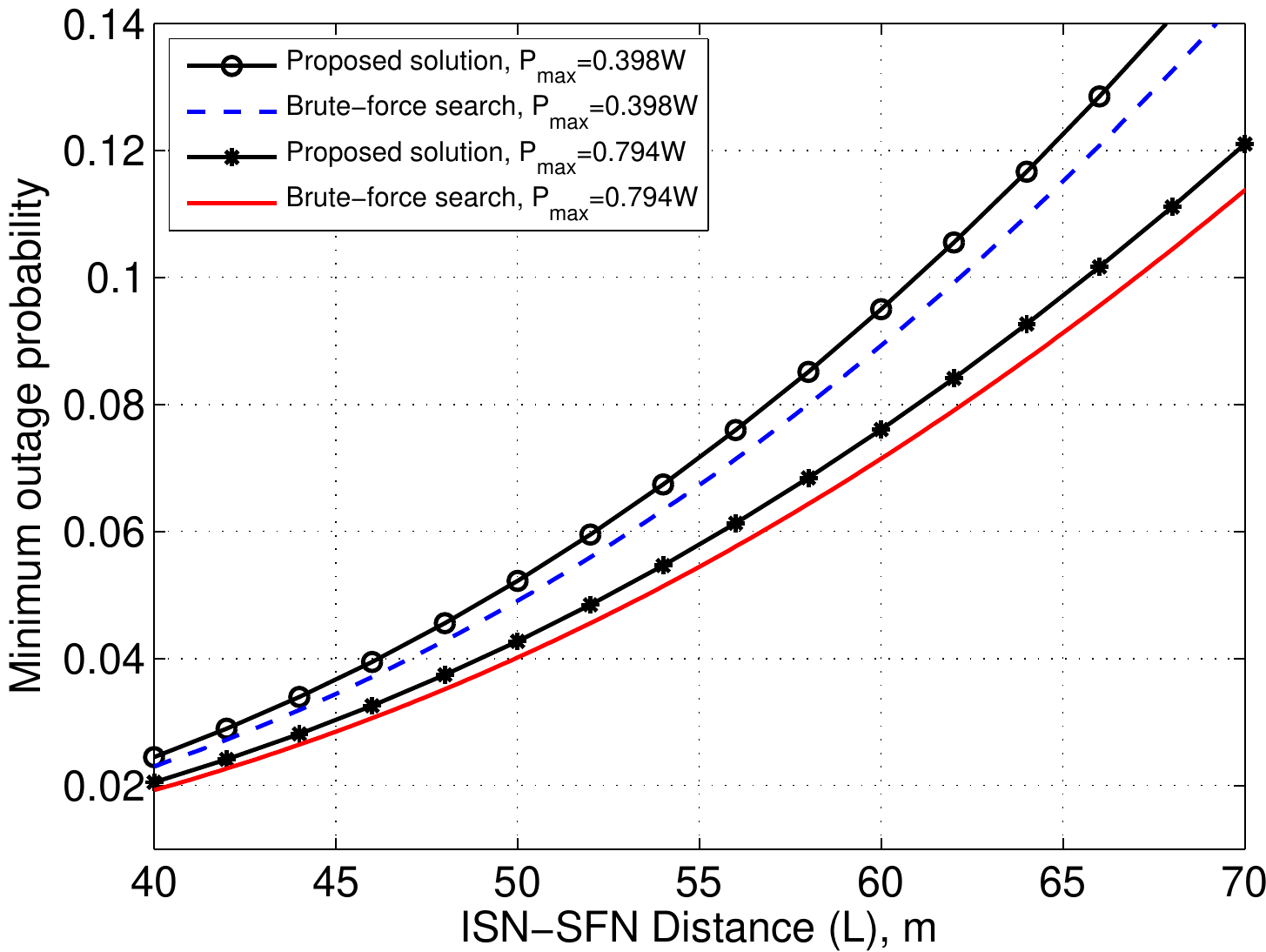}
\caption{ISN-SFN Distance,~$L$ vs. Minimum outage probability.}
\label{relaydist}
\end{figure}

Fig. \ref{relaypower} show the minimum outage probability when plotted against the transmit power of the RFN. Since our proposed solution used significantly large number of transmission time slots,~$K$ for the RFN to achieve the minimum outage probability, we employed the brute-force searching algorithm to obtain the outage probability using about 800 power-control and 3000 RFN location possibilities. We observe that the brute-force search which systematically enumerates all possible candidates for the solution and checks whether each candidate satisfies the constraints, yielded a performance improvement of about 7\% when compared with the proposed solution. It can also be seen that as~$P_{R}^{t+1} \rightarrow 0$, the outage becomes certain, likewise, as~$P_{R}^{t+1}\rightarrow P_{max}$ the outage also becomes certain. In essence, we observe convexity of the outage probability function with respect to~$P_{R}^{t+1}$. From the results in Fig. \ref{relaypower}, we can obtain the value of~$P_{R}^{t+1}$ that gives the minimum outage, that value is regarded as the optimal transmit power at that given separation~$L$ between the source and destination. Furthermore, we also observe that as the separation between the ISN and the SFN drops from 50~$m$ to 48~$m$, the outage probability drops, thus, this implies higher outages for larger values of~$L$. Fig. \ref{relaydist} illustrates the behavior of the minimum outage probability when the separation between the ISN and the SFN changes. We observe a monotonic increase in the outage probability as~$L$ increases. This implies that distance between the source ISN and destination becomes large, higher communication outage will be experienced, thus, validating the results in Fig. \ref{relaypower}. Similarly, we observe that the brute-force search yielded about 6\% performance improvement.

\begin{figure}[!t]
\centering
\includegraphics[width=3.2in]{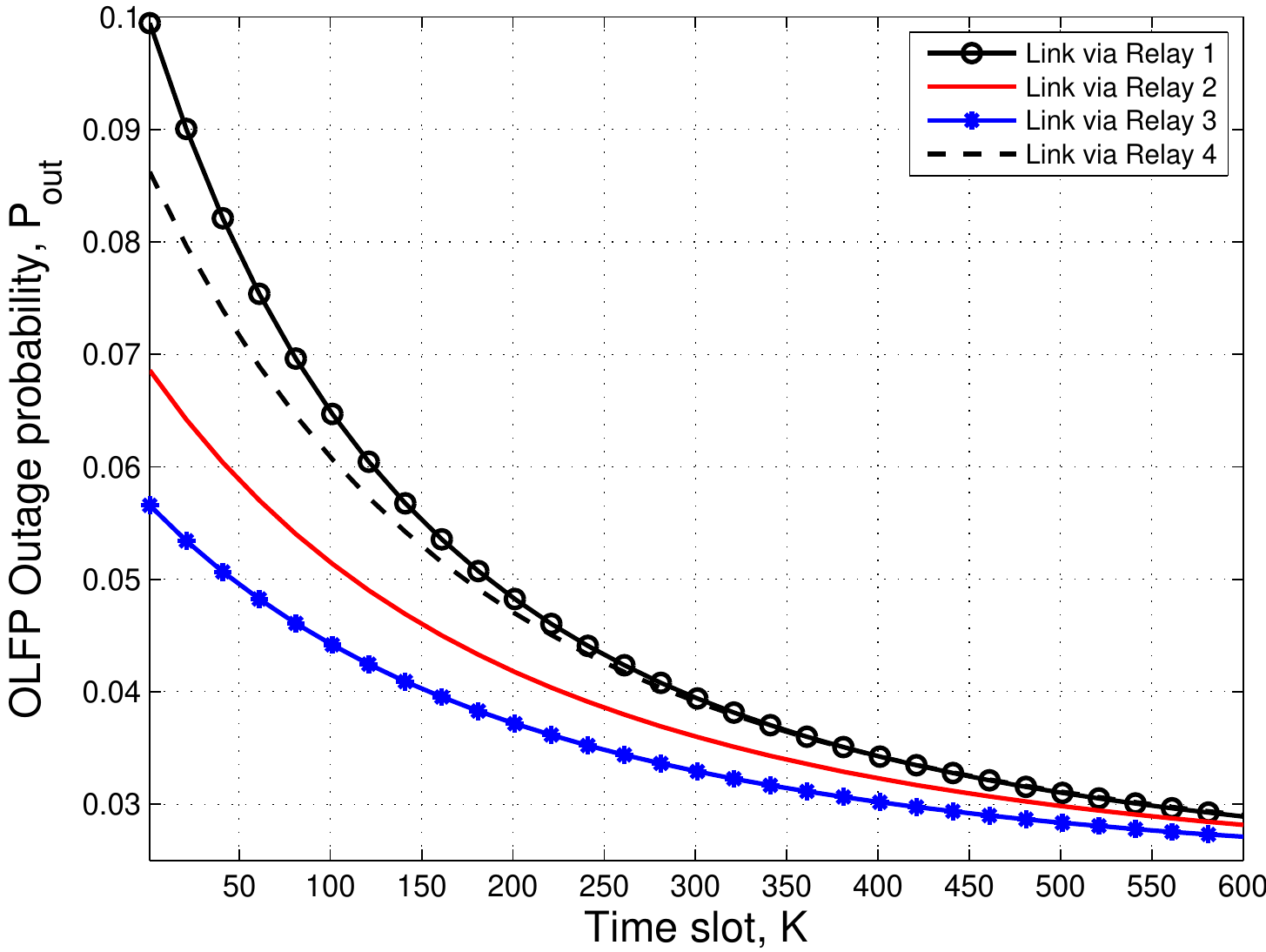}
\caption{Time slot,~$K$ vs. Outage probability, $\mathcal{P}_{out}$ using the OLFP scheme.}
\label{olfp}
\end{figure}

In Fig.~\ref{olfp},~Fig.~\ref{opfl}~and Fig.~\ref{olop}, we examine the link behavior of the proposed schemes using a four (4) relay scenario. Results in Fig.~\ref{olop} show the best convergence with most minimal outage as compared to the OLFP and OPFL schemes. However, we observe that the initial convergence rate of all four relays is closely matched. The same can be seen in Fig.~\ref{opfl} where all relays have almost same initial convergence rate. On the contrary, Fig.~\ref{olfp} show quite unique, but poor initial convergence rate for each relay. This may be as a result of the adjustment of its position to minimize outage. In general, at the time slot,~$K = 600$ all the four relays can be seen to converge in each of the schemes. The choice of the scheme to be used may depend on the fog-based IoT design requirements. This may also depend on the computational complexity of the chosen scheme.

\begin{figure}[!t]
\centering
\includegraphics[width=3.2in]{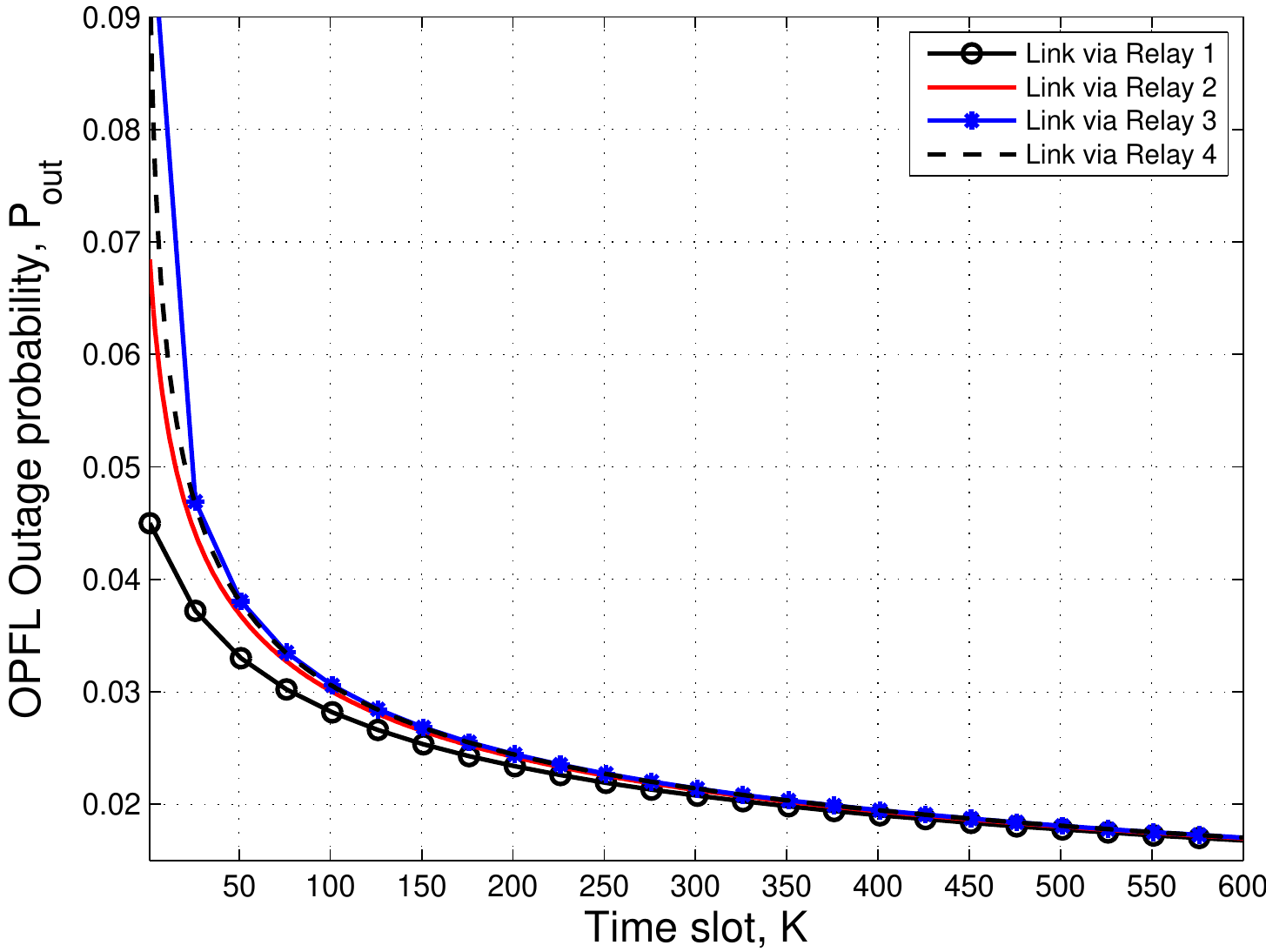}
\caption{Time slot,~$K$ vs. Outage probability, $\mathcal{P}_{out}$ using the OPFL scheme.}
\label{opfl}
\end{figure}

\begin{figure}[!t]
\centering
\includegraphics[width=3.2in]{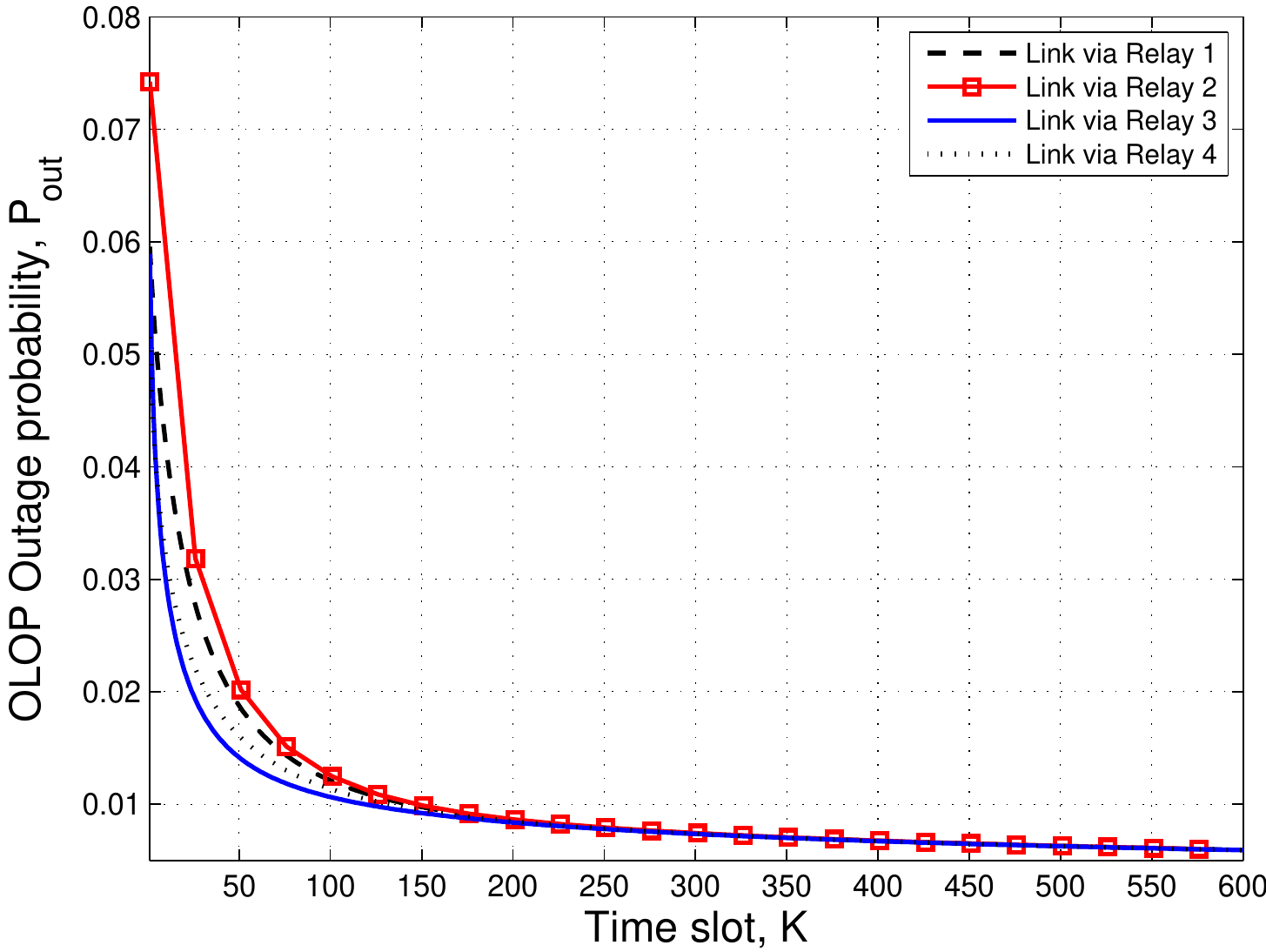}
\caption{Time slot,~$K$ vs. Outage probability, $\mathcal{P}_{out}$ using the OLOP scheme.}
\label{olop}
\end{figure}

\begin{figure}[!t]
\centering
\includegraphics[width=3.2in]{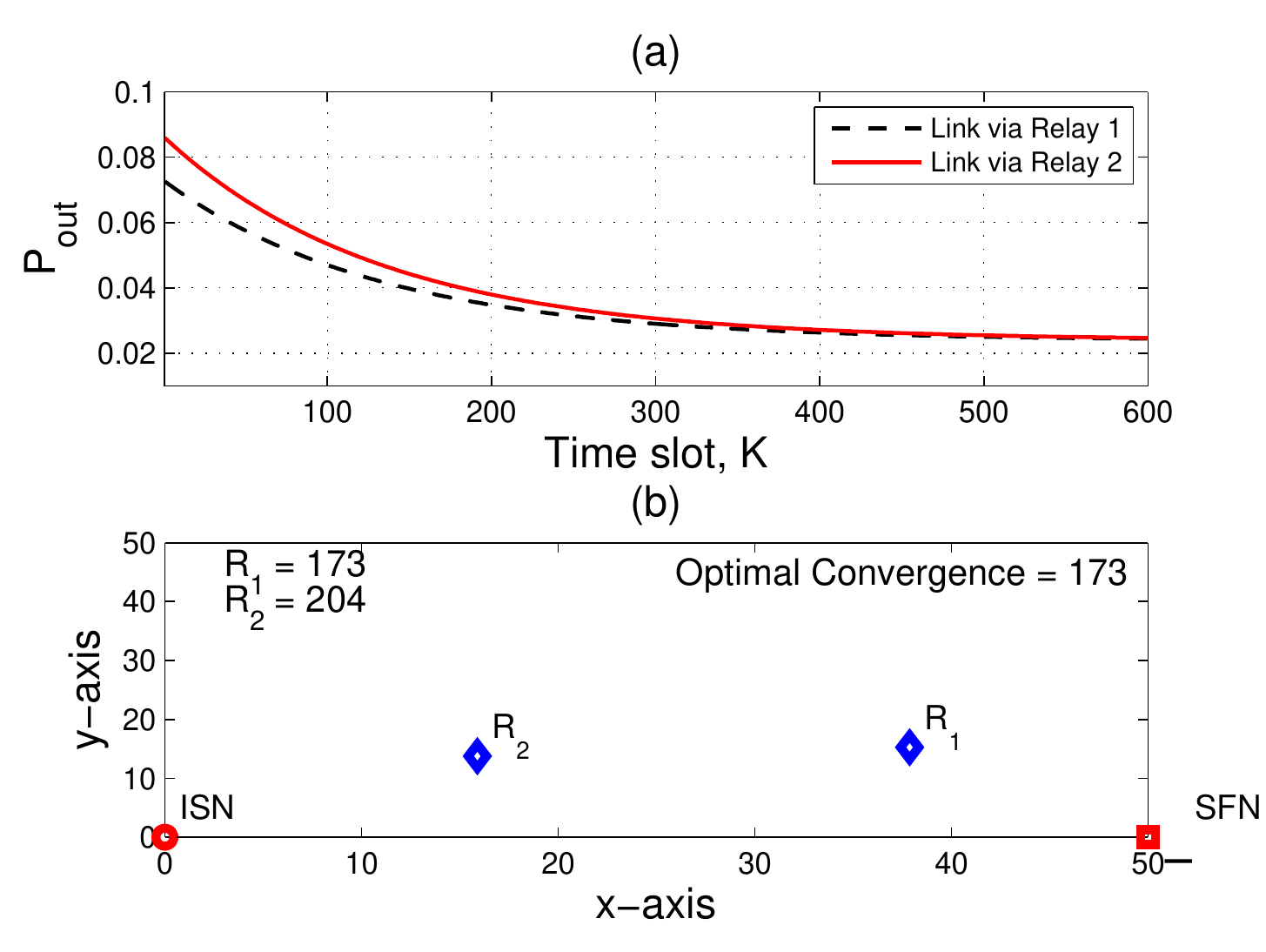}
\caption{Two Relay Scenario (a) Convergence behaviour for each link using the OLFP scheme, (b) Optimal selection of randomly deployed fog relays.}
\label{doubleRelay}
\end{figure}
In order to illustrate our proposed relay selection strategy, we have chosen the OLPF scheme with randomly deployed fog relays.
Fig. \ref{doubleRelay} (a) show the convergence behaviour for each link in a two relay scenario. We modeled the region of potential relay candidates to be a semicircle of diameter 50~\emph{m}, between the ISN and SFN. The relays are then randomly deployed within the region of potential relay candidates as shown in Fig. \ref{systemmodel2}. We then examined the performance of each relay using the convergence criteria. We can see in Fig. \ref{doubleRelay} (b) that~$R_1$ converged faster than~$R_2$ in that particular scenario. This implies that the link via~$R_1$ will be selected during that transmission phase. The proposed algorithm for relay selection can be further applied to ensure fairness and energy balance within the system. Similarly, Fig. \ref{threerelay} (a) and  Fig. \ref{fourRelay} (a) show the convergence behaviour for three and four relay scenario, respectively. The random deployment of these scenarios can be seen in Fig. \ref{threerelay} (b) and  Fig. \ref{fourRelay} (b), with the optimal relay selected in each case. From Fig. \ref{doubleRelay}, Fig. \ref{threerelay} and Fig. \ref{fourRelay}, we observe that the RFN closer to SFN is selected, which may be due to smaller communication outage experienced in the~$(t + 1)$ time slot when the optimal relay is used. The proposed strategy can be applied to large scale fog-based networks where multiple fog devices could be used as potential relays to forward traffic from source to destination.

\begin{figure}[!t]
\centering
\includegraphics[width=3.2in]{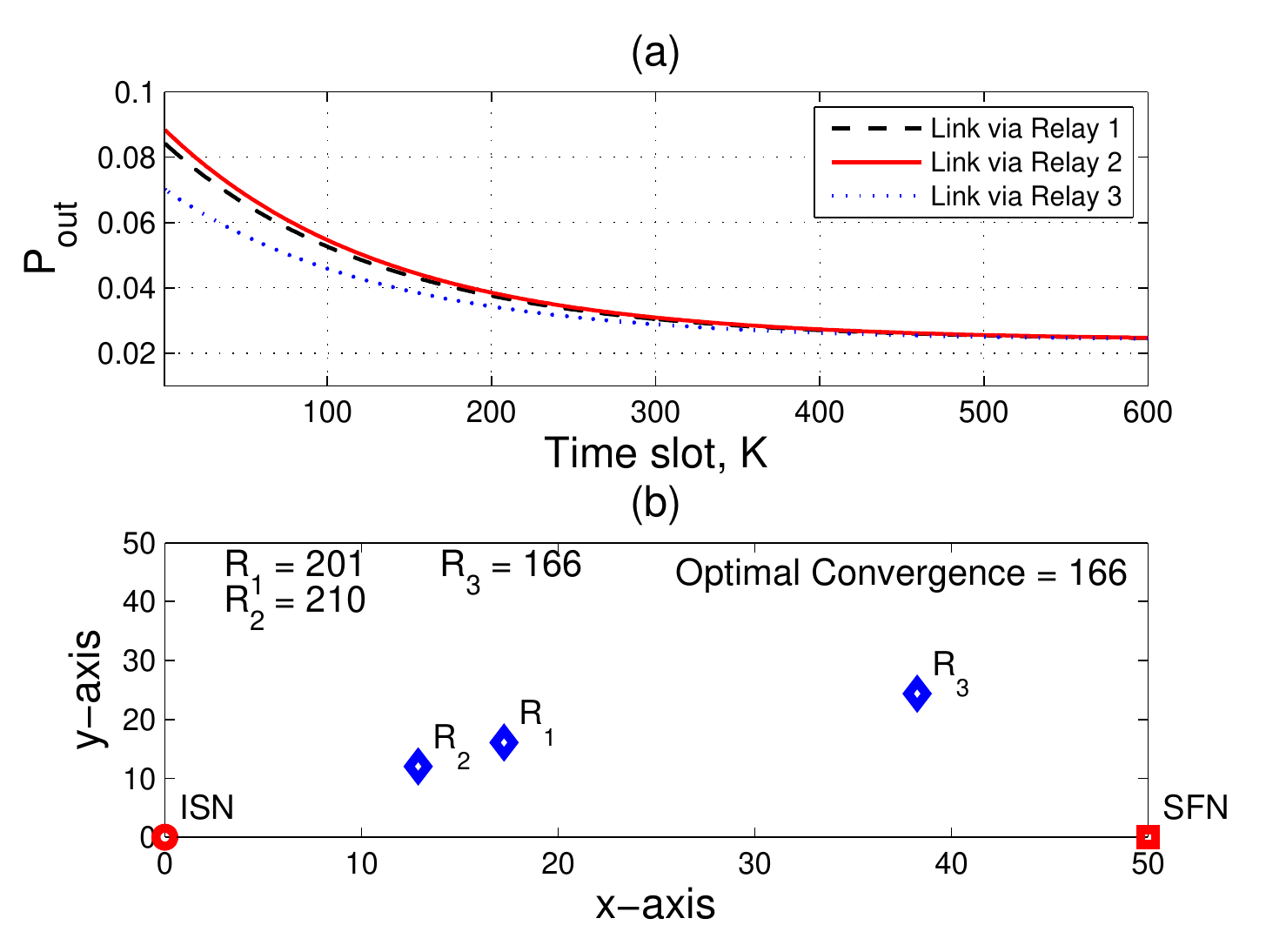}
\caption{Three Relay Scenario (a) Convergence behaviour for each link using the OLFP scheme, (b) Optimal selection of randomly deployed fog relays.}
\label{threerelay}
\end{figure}

\begin{figure}[!t]
\centering
\includegraphics[width=3.2in]{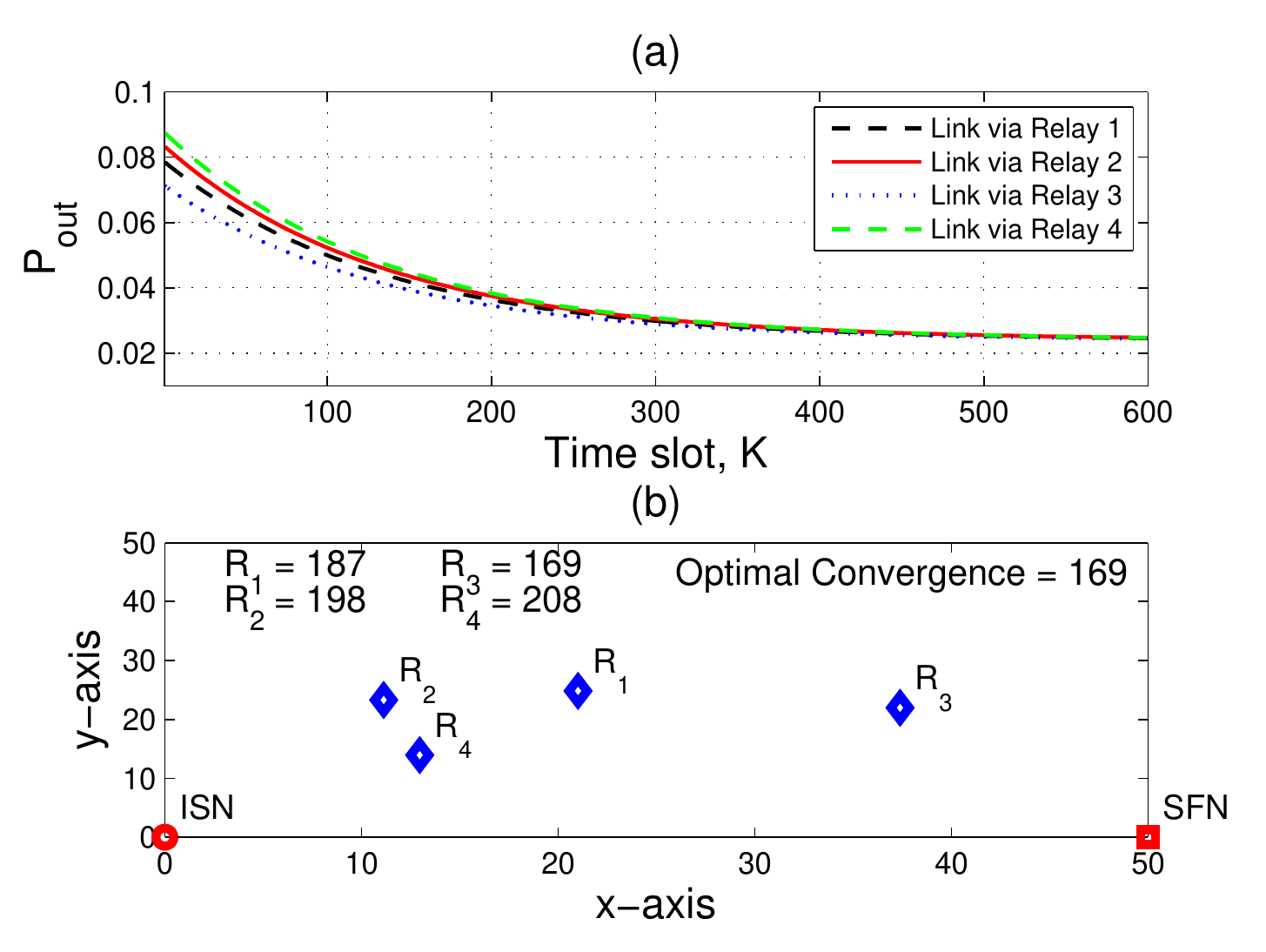}
\caption{Four Relay Scenario (a) Convergence behaviour for each link using the OLFP scheme, (b) Optimal selection of randomly deployed fog relays.}
\label{fourRelay}
\end{figure}

\section{Conclusion}
\label{conclusion}
In this paper, we analyzed the outage performance of relays in a fog-based IoT Network.
Firstly, the outage probability of the presented system model was derived,
leading to the formulation of the non-convex optimization problem. We presented three design schemes using the steepest descent method to ensure that the outage probability is minimized. Simulations reveal that the joint location and transmit power-control optimization of the relay significantly minimized the communication outage within the fog-based IoT network.
Moreover, since the proposed scheme can be readily deployed, the system designer may opt for any of the proposed schemes, depending on the design requirements. Secondly, an optimal relay selection algorithm was proposed to select an optimal link using the convergence criteria from potential relay candidates. This proposed strategy will not only select an optimal link, but will ensure fairness and longevity of the power-constrained fog relay. Our future work will consider incorporating the proposed mechanism into a MAC protocol for fog-based IoT networks.

\appendices
\section{Proof of Lemma 1}
From (\ref{eqn8}), we further simplify the expression~$\gamma_t$, given by
\begin{equation}\label{AAeqn1}
 \gamma_t = \frac{c P_{R\leftarrow I}^{t} P_{S\leftarrow R}^{t + 1}}{c P_{S\leftarrow R}^{t + 1} N_0 + N_0}
\end{equation}
where~$c = G ^2 / P_R^{t + 1}$. We have that~$P_{R\leftarrow I}^{t}$~and~$P_{S\leftarrow R}^{t + 1}$ follows an exponential distribution
based on the physical characteristics of the model considered.
According to (\ref{eqn8}) and (\ref{AAeqn1}), the predefined threshold of the outage can be expressed in terms of some independent variables~$Q$~and~$Z$,
\begin{equation}\label{AAeqn2}
 \mathcal{P}_{out}^t = \mathcal{P}(P_{R\leftarrow I}^{t} \leq Q + Z)
\end{equation}
where~$Q = N_0 \tilde{\gamma}$~and~$Z = N_0 \tilde{\gamma}/ c P_{S\leftarrow R}^{t + 1}$. Let~$\Upsilon = P_I^t (D_I^t)^{-\alpha}$,
we now get the outage probability, expressed as the average value over all realizations of~$\mathcal{P}(P_{R\leftarrow I}^{t} \leq Q + Z)$, as
\begin{equation}\label{AAeqn3}
\begin{split}
 \mathcal{P}_{out}^t & =  \mathop{\mathbb{E}}_{Q + Z} \{\mathcal{P} (P_{R \leftarrow I}^{t} \leq q + z | q + z)\} \\
 & = \mathop{\mathbb{E}}_{Q + Z} \Big\{ \int_0^{q + z} \frac{\exp{\Big(-x/\Upsilon \Big)}}{\Upsilon}dx \Big\}\\
 & = \mathop{\mathbb{E}}_{Q + Z} \Big\{ 1 - \exp\Big(\frac{-(q + z)}{\Upsilon}\Big) \Big\}\\
& = 1 - \exp\Big( \frac{-(N_0 \tilde{\gamma})}{\Upsilon} \Big) \cdot \mathop{\mathbb{E}}_Z \Big\{\exp\Big( \frac{-Z}{\Upsilon} \Big)\Big\}
 \end{split}
\end{equation}
In order to find the expected value of the function of~$Z$, we first find the pdf of~$Z$. We note that since~$Z = N_0 \tilde{\gamma}/ c P_{S\leftarrow R}^{t + 1}$,
we can find the pdf of~$Z$ by first finding the pdf of~$P_{S\leftarrow R}^{t + 1}$. This can be seen in the transformation
\begin{equation}\label{AAeqn4}
 f_Z(z) = \frac{N_0 \tilde{\gamma}}{c z^2} \cdot \frac{1}{\sigma} \exp{\Big(-\frac{N_0 \tilde{\gamma}}{\sigma cz}\Big)}
\end{equation}
where~$\sigma = P_{R}^{t + 1} (D_{S}^{t + 1})^{-\alpha}$, the outage probability can be expressed as

\begin{equation}\label{AAeqn5}
\begin{split}
 \mathcal{P}_{out}^t &= 1 - \frac{1}{\sigma}\exp{\Big(-\frac{N_0 \tilde{\gamma}}{\Upsilon}\Big)}                                                \\
           & \times \int_{0}^{+\infty} \exp{\Big( \frac{z}{\Upsilon} \Big)} \cdot  \exp{\Big(-\frac{N_0 \tilde{\gamma}}{\sigma cz}\Big)} \cdot \Big(\frac{N_0 \tilde{\gamma}}{cz^2}\Big) dz
\end{split}
\end{equation}

\begin{equation}\label{AAeqn6}
\begin{split}
 \mathcal{P}_{out}^t &= 1 - \exp{\Big(-\frac{N_0 \tilde{\gamma}}{\Upsilon}\Big)} \cdot 2\Psi \mathbb{B}_{-1}(2\Psi)                                                \\
          \Psi &= \sqrt{\frac{N_0 \tilde{\gamma} (\Upsilon + N_0)}{\sigma \Upsilon}}
\end{split}
\end{equation}
where~$\mathbb{B}_{-1}$ is the first order negative modified Bessel function of the second type,
which has the property~$\mathbb{B}_{-1} = ~\mathbb{B}_{1}$, expandable as~$\mathbb{B}_{1}(x)~\simeq~ 1/x + (x/2)ln(x/2)$ ~\cite{Luke1962}.
Hence, the analytical approximation for the outage probability~$\mathcal{P}_{out}^t$ is given in (\ref{eqn10}). This completes
the proof.

\section{Proof of Lemma 2}
If~$\lambda$ is chosen such that~$\mathcal{P}_{out}(x_i + \lambda h_i)$ is minimized in each iteration,
then successive directions are orthogonal~\cite{Wang2008}. Hence, we show that
\begin{equation}\label{ABeqn1}
\begin{split}
  \frac{d\mathcal{P}_{out}(x_i + \lambda h_i)}{d \lambda} &= \sum_{j=1}^{n} \frac{\partial \mathcal{P}_{out}(x_i + \lambda h_i)}{\partial x_{ij}} \frac{(x_i + \lambda h_i)}{d \lambda}\\
&= \sum_{j=1}^{n} \nabla_j (x_i + \lambda h_i) h_{ij}\\
&=   \nabla (x_i + \lambda h_i)^T
\end{split}
\end{equation}
where~$\nabla (x_i + \lambda h_i)$ is the gradient at point~$x_i + \lambda h_i$.
If~$\hat{\lambda}$ is the value of~$\lambda$ that minimizes~$\mathcal{P}_{out}(x_i + \lambda h_i)$ then
\begin{equation}\label{ABeqn2}
\begin{split}
  &\nabla (x_i + \hat{\lambda} h_i)^T h_{i} = 0\\
 ~~&~~or\\
  &  h_{i + 1}^T \cdot h_i = 0
\end{split}
\end{equation}
where the steepest-descent direction at the point~$x_i + \hat{\lambda }h_i$ is given by
 \begin{equation}\label{ABeqn3}
 h_{i + 1} = -\nabla (x_i + \hat{\lambda} h_i)
\end{equation}
As such, successive directions~$h_i$ and~$h_{i + 1}$ are orthogonal. This completes
the proof.



\ifCLASSOPTIONcaptionsoff
  \newpage
\fi

\end{document}